\documentclass{sig-alternate-2013}

\usepackage{epsf}
\usepackage{color}
\usepackage{epsfig}
\usepackage{graphicx}
\usepackage{subfigure}
\usepackage{amsmath,amssymb}
\usepackage{algorithm}
\usepackage[noend]{algorithmic}
\usepackage{multirow}
\usepackage{xspace}
\usepackage{url}
\usepackage[normalem]{ulem}

\newcommand{\red}[1]{\textcolor{black}{#1}}
\newcommand{\blue}[1]{\textcolor{black}{#1}}

\newcommand{\erase}[1]{}

\def\0{{\mathbf 0}}
\def\1{{\mathbf 1}}

\newcommand{\topk}{top-$k$\xspace}

\newcommand{\len}{\mbox {$\ell$}\xspace}
\newcommand{\DP}{Dynamic Programing\xspace}
\newcommand{\Rank}{SAA\xspace}
\newcommand{\ClassicRJ}{{Rank Join}\xspace}
\newcommand{\RJ}{{\sc NextHeavyPath}\xspace}
\newcommand{\Main}{{\sc RSA}\xspace}

\newcommand{\hpp}{HPP\xspace}

\newcommand{\nop}[1]{}

\newcommand{\eat}[1]{}
\newcommand{\para}[1]{\medskip\noindent {\bf #1.}}
\newcommand{\wmax}{\mbox {$w_{max}$}}

\newtheorem{myprop}{\textbf{PROPOSITION}}
\newtheorem{mytheo}{\textbf{THEOREM}}

\newtheorem{mylem}{\textbf{LEMMA}}
\newtheorem{myfact}{\textbf{FACT}}
\newtheorem{myobs}{\textbf{LIMITATION}}
\newtheorem{myopt}{\textbf{OPTIMIZATION}}

\newlength{\figwidth}
\newlength{\figthree}
\newlength{\figfour}
\newlength{\figfours}
\newfont{\mycrnotice}{ptmr8t at 7pt}
 \newfont{\myconfname}{ptmri8t at 7pt}
 \let\crnotice\mycrnotice%
 \let\confname\myconfname%

\begin{document}

\title{Finding HeavyPaths in Weighted Graphs and a Case-Study on Community Detection}

\author{
Mohammad Khabbaz\\
       \affaddr{Personal Business Development}\\
       \email{mohammmad@gmail.com}
}

\maketitle

\begin{abstract}

A heavy path in a weighted graph represents a notion of connectivity and ordering that goes beyond
two nodes. The heaviest path of length $\ell$ in the graph, simply means a sequence of nodes with edges between them,
such that the sum of edge weights is maximum among all paths of length $\ell$. It is trivial to state the heaviest edge in the graph is the heaviest path of length $1$, that
represents a heavy connection between (any) two existing nodes. This can be generalized
in many different ways for more than two nodes, one of which is finding the heavy weight paths in the graph.
In an influence network, this represents a highway for spreading information from a
node to one of its indirect neighbors at distance $\ell$. Moreover, a heavy path implies
an ordering of nodes. For instance,
we can discover which ordering of songs (tourist spots) on a playlist (travel itinerary) is
more pleasant to a user or a group of users who enjoy all songs (tourist spots) on
the playlist (itinerary). This can also serve as a hard optimization problem,
maximizing different types of quantities of a path such as score, flow, probability
or surprise, defined as edge weight. Therefore, if
one can solve the Heavy Path Problem (HPP) efficiently, they can as well use HPP for modeling
and reduce other complex problems to it.

More precisely, we aim at finding $k$ heaviest (top-$k$) paths of a given
length $\ell$. The weight of a path is defined as a monotone
aggregation of individual edge weights using functions such as sum or product.
We argue finding simple paths is way more practical than finding
paths with cycles in applications such as: routing, playlist recommendation, itinerary planning
and influence maximization. Avoiding cycles results in lists without repeated
items and with better diversity. Simple paths are also expected to have a higher utility in practice.
This makes \hpp NP-hard and inapproximable. We propose an efficient algorithm
that despite its \emph{exponential} (theoretical) \emph{worst case} running time, achieves the
exact answer of the NP-hard problem in many useful cases
and study the problem from different perspectives.

We compare our main algorithm, Repeated Sorted Access (RSA), against baseline and state-of-the-art
algorithms. We show with experiments
that following our approach is significantly more scalable than
existing algorithms for solving \hpp.
We conduct a comprehensive set of experiments
on four graphs that inherit characteristics of real life applications.
This supports our arguments regarding
scalability of \Main in practice. We make all of our implementations as well as
graphs publicly available to enable reproducibility of our experiments. Our focus
is on solving the technical problem in the first segment of the paper. Following
our technical presentation, we provide a novel case study on community detection
to illustrate the usefulness of solving HPP with a novel application.

\end{abstract}
\vspace{-10pt}
\section{Introduction}
\label{sec:intro}
Graphs are
convenient data structures for modeling real life data such as
social networks, citation networks, the web, phone call and email histories,
purchase/transaction records and user actions~\cite{Broder:2000}.
It is easy to model data with graphs. However, we
need sophisticated tools to make sense
of such data. In order to do so, we need to
study notions such as connectivity~\cite{Broder:2000},
centrality~\cite{Kang:11} and density~\cite{Gibson:05}. Graph
algorithms are at the core of community detection, for
discovering various types of groups based on previous interactions
of individuals~\cite{Santo:09}. Random walk based methods have been
used in many mining and graph applications.
Mining interesting multi-relational patterns is solved as discovering maximal
pseudo cliques in k-partite graphs~\cite{Spyropoulou:09}.

A problem that has received relatively less attention in the \emph{research
community} is that of finding $k$ cycle-free heaviest paths
of length $\ell$ in a weighted graph (\hpp). Sum and
product are examples of monotone aggregation functions that define
the aggregate weight of a path.
\eat{
We further assume edge weights are either all positive or all negative~\footnote{Negative weights help with transforming probabilities to log values and converting product to sum. In this case all the edge weights are negative since they
are logs of numeric values between $0$ and $1$.}.
}
Generating lists of
items of interest is a great practical example. \eat{In such applications
we do not mind where the list starts or ends. We assume all items are
of interest, and are looking for a list of highly connected items.}
Suppose edge probability $w_{(u,v)} = P(v|u)$, the conditional
transition probability that a random user would like to listen to $u$
(as next song) given that $v$ is the current song.
Also assume there is equal probability for any user to like any song.
\eat{This simply means the database
of songs consists of only the "good" songs. We can generalize the model
to incorporate node probabilities proportional to user interests as well.
In this paper, we assume equal probability for all nodes for the sake
of the simplicity of our technical presentation.}
Using product as aggregation function,
the weight of a path is directly proportional to the probability that a cycle free \emph{"chain"}
of songs appears one after another in a listening session. A heavy path in such
graph, represents a playlist with high overall transition probability between songs.
The graph can be created according to user(s) listening history. We can imagine other kinds
of lists such as reading lists of books or papers that show flows of ideas through
publications.

In influence networks, recently there have been attempts for sparcification~\cite{Sparsification:11}.
A path with high probability is a highway for spreading influence in an influence network
with transition probabilities on edges. A more principled approach takes into account
only heavy paths rather than technicalities and complex problem definitions. A complete study of all of \hpp applications goes beyond
the scope of this paper. We need the preliminary work done in this paper,
in order to aim at creating a toolbox that uses heavy paths for network analysis.
Here, we show a possible novel application for
community detection in research communities in a co-authorship network. We use maximal frequent sub-paths in top-$k$
heavy paths in order to discover the \emph{cores} and further \emph{"grow"} each core with
peripheral nodes using existing edges in top paths. We enhance our observations with meta-data, and show that
the core communities found using our approach unanimously report highly influential researchers.

We make the following contributions in the subsequent sections:
\begin{itemize}
\item We formalize and motivate the \hpp problem and show it is
NP-hard and inapproximable. We discuss four exact algorithms, and
provide a comparison of classical vs. novel algorithms for solving HPP.
We provide a comprehensive experimental study of HPP algorithms presented in Section~\ref{sec:expn}.

\item We present a case study to discover small but influential networks of people in research communities.
We do this analysis on a sub-graph of DBLP co-authorship graph, and report the results
in Section~\ref{sec:community}.

\item We show how a well studied top-$k$ query processing problem in the database
 literature can be converted to HPP on graphs where there is no database. We solely focus
 on algorithmic aspects, and propose a more scalable algorithm that overcomes the shortcommings
 of the Rankjoin~\cite{Ilyas:04} algorithm.

\end{itemize}


\section{Technical Problem Definition}
\label{sec:probdef}
\begin{figure*}[t]
  \centering
  \subfigure[]{
  \includegraphics[scale = 0.7, trim = 0mm 0mm 0mm 20mm]{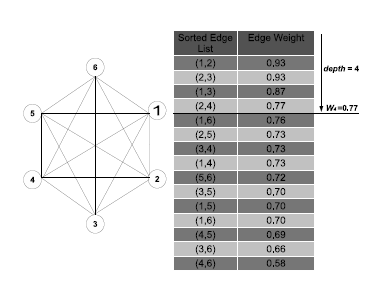}
  \vspace{-10pt}
  \label{fig:cliqueGraph}
}
\subfigure[]{
\includegraphics[scale = 0.7, trim = 0mm 0mm 0mm 50mm]{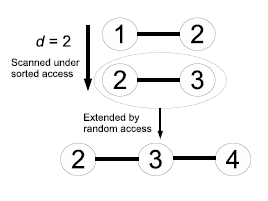}
\label{fig:depth2}
}\label{fig:access}
\vspace{-10pt}
\caption{\scriptsize{(a) A complete graph with 6 nodes and sorted list of edges. (b) An example of random access is shown on the right.}}
\vspace{-10pt}
\end{figure*}
The most similar well studied problem to HPP in the literature is finding top-$k$
heavy paths on directed \len-partite graphs.
A nice example is \emph{$k\ell$-Stable Clusters} problem~\cite{Bansal:07}.
Absence of cycles makes the studied problem considerably simpler than HPP.
Furthermore, practical
applications of HPP are not limited to \len-partite graphs as described
in Section~\ref{sec:intro}. We show
HPP is NP-hard and propose a practical algorithmic framework
for solving it in the general case. It is common to use exponential time but efficient
algorithms for mining purposes. Great examples are \textbf{\emph{Apriori}} and \textbf{\emph{FP-Growth}}~\cite{FPGrowth}, that is
used for frequent item-set mining. Heavy paths also represent lists of highly connected items
that can be used in similar ways to frequent item-sets and also for ordering of the items in the set.
\eat{
Frequent item-set is a stronger notion than a clique, therefore the frequent item-item graph
will contain a clique for sure. A heavy path that goes through the nodes of the clique
can also order items in each frequent item-set.
}

Given a weighted graph $G(V,E,W)$, where
weight $w_{(u,v)}$ represents the edge weight between $u$ and $v$,
we aim at finding the $k$ heaviest simple paths of length $\len$, with
highest overall path weights. We call this the {\em Heavy Path Problem (\hpp)}.

A simple path of length \len is a sequence of
nodes $(v_0, \ldots,
v_{\ell})$ such that $(v_i, v_{i+1})\in E$, $0\leq i < \len$, and there are no
cycles, i.e., the nodes $v_i$ are distinct.
Unless otherwise specified, in the rest of the paper, we use the term
path to mean simple path. We focus on sum and product, as aggregation
functions that define the weight of a path. For simplicity of presentation and practical considerations, we present using sum. We define weight of $P = (v_0, ..., v_\ell)$, as
$P.weight = \sum_{j=0}^{\ell-1} w_{(j,j+1)}$. Furthermore, we know we can
convert product to sum using log values and this does not change the order
of the list of top-$k$ heaviest paths.

Finding the heaviest path for a given parameter \len is equivalent to the \len-TSP
problem, defined as follows: Given a weighted graph, find a path of \textbf{\emph{minimum}} weight that passes through
any $\len+1$ nodes. It is easy to see that for a given length \len,
a path $P$ is a solution to \len-TSP
iff it is a solution to \hpp (with $k=1$) on the same graph but with
edge weights \erase{appropriately} modified as follows: let
$w_{(u,v)}$ be the weight of an edge $(u,v)$ in the \len-TSP instance;
then the edge weight in the \hpp instance is
$w'_{(u,v)}=C-w_{(u,v)}$, where $C$ is any constant.
This reduction works in both ways given that we use the same constant $C$.
We are not
concerned with negative edge weights since we deal with paths without cycles; but, we
focus on cases where edges are either all negative or all positive for simplicity.
{\sl It is well known that \len-TSP is NP-hard (as a version of TSP)}. Furthermore, \len-TSP
does not have any bounded approximation in the general case i.e. no triangle inequality on
edge weights etc~\cite{mst1, mst2}.
\begin{myprop}
HPP does not have any bounded approximation in polynomial time.
\end{myprop}
\begin{proof}
\scriptsize{
We know \len-TSP is not approximable. Remember the $C - w$ reduction between
\len-TSP($G1$) and \hpp($G2$) introduced earlier. Suppose $C = 0$ and $k = 1$, if $P$ is
the answer to \hpp (the heaviest path) on $G2$, it is the answer to \len-TSP on $G1$
after $-w$ transformation on edge weights. The weight of $P$ in $G1$ ($-W_{max}$) is the negative of its
weight in $G2$ ($W_{max}$). If we find an $\alpha$ ($0<\alpha<1$) approximation to $P$
on $G2$ as a result of finding path $Q$ in polynomial time we have $Q.W \ge \alpha \times W_{max}$.
On $G1$ graph equivalently we can say, $-Q.W \le \alpha \times -W_{max}$. Since $-W_{max}$ is the
optimal answer on $G1$ graph (i.e. minimum weight among all paths), it is guaranteed that $-Q.W$ is
in the range $-W_{max} \le -Q.W \le \alpha \times -W_{max}$ and is a true approximate answer to \len-TSP.
We know \len-TSP is not approximable. According to the above argument, approximability of \hpp results in approximability of
\len-TSP. This is in contradiction with the known fact that \len-TSP and TSP are not approximable in
the general case.
}
\end{proof}

In Figure~\ref{fig:cliqueGraph}, the heaviest simple path of length $4$,
is obtained by visiting the nodes in the order 6--1--2--3--4 and has a weight of 3.35.
In contrast, a different permutation of the nodes 4--3--6--1--2 has a weight of 3.08. The higher
weight of a path can represent a sequence with more strength or better quality. Different
permutations result in paths with different weights simply because paths consist of different edges most
likely with different edge weights.

\section{Algorithms}
\label{sec:algo}
\subsection{Baselines}
An obvious algorithm for finding the heaviest paths of length \len
is performing a depth-first search (DFS) from each node, with the search
limited to a depth of \len, while maintaining the
$k$ heaviest paths and avoiding cycles. This is an exhaustive algorithm and is not expected to scale.
A somewhat better approach is dynamic programming. Held and Karp~\cite{Held:62} proposed a
dynamic programming algorithm for TSP which we adapt to HPP as follows. Since path lengths
required in our problem are typically much smaller than the total number of nodes in the graph,
we replace the notion of "allowed nodes" with "avoiding nodes". The idea is to find the heaviest
$J$-avoiding path of length $\ell-1$ that ends in one of the neighbors of $J$. We can repeat
this recursively during the execution using smaller path lengths and larger avoidance sets as parameters.
The base case is where the input parameter assigned to length is $1$, in which we only
need to consider the paths created by extending the current path by one hop, making
sure we do not extend by any nodes that belong to the ($\ell - 1$) size avoidance set. The rest
is similar to DFS. We repeat this for every node($J$) and keep track of the heaviest path. Dynamic
Programming algorithm aggregates and discards many short path segments early on by choosing
the heaviest one that ends at a certain node. Both DFS and Dynamic Programming can be easily generalized
to report top-$k$ paths which is not the focus of this paper.
\subsection{Sorted Access Algorithm (SAA) for \hpp}
\label{sec:RJ}
The methods discussed in the previous section lack the ability to prune
the search space using edge weights. It may be more promising to first
look at those parts of the graph with heavier edge weights. If they
connect, they can create heavy paths. This way, we solve HPP only for a
smaller sub-graph or at least explore the search space more efficiently. This
can make a considerable difference to the exponentially increasing running time
of the problem using smaller input size.
There exists a body of work on \topk algorithms in
the database community started by Fagin et al.~\cite{TA}. These algorithms
aim at finding top-$k$ items without searching the whole search space
using \emph{sorted access}. The framework is such that partial item scores
from multiple sorted lists ($l_1 ... l_{\ell}$), are aggregated
using a monotone aggregation function. Following their idea more aggressive \topk algorithms were
proposed in order to extend their framework with probabilistic guarantees such as~\cite{topk-prob}.
We can argue the weight of a path (item), is the aggregation of several partial
item scores (edges). This makes following top-$k$ query processing ideas
promising for solving HPP. In particular, in \ClassicRJ~\cite{Ilyas:03,Schnaitter:10} problem,
partial scores in $l_1$ can join more than one partial scores in $l_2$.
For instance, suppose $l_1$
provides a sorted list of movie theaters and $l_2$ provides
a sorted list of restaurants sorted by their popularity score.
One may be interested in watching a movie and going to a restaurant
in the same neighborhood. The score of a combination is defined using the sum of
the individual scores and items of interest are (movie theater, restaurant) combinations.
There may be
many movie theaters and restaurants in the same neighborhood.

\textbf{\emph{Problem Conversion:}} \emph{Define a tripartite graph. The first column of nodes is a single dummy node (dummy1) as well as the third column(dummy3). Every node in the middle column represents a neighborhood (such as downtown).
Every movie theater ($l_1$), is modeled as an edge that connects dummy1 to its neighborhood
with a directed edge. Every restaurant ($l_2$), is modeled as a directed edge that connects
its neighborhood to dummy3. Edge weight is defined as item score. Heaviest paths of length
$2$ are top combinations. Rankjoin can be converted to heavy path on a tripartite graph.}

We design an adaptation of Rankjoin that works with self-joins (i.e. all $l_i$ are identical) for solving
\hpp on general graphs. Direct usage of Rankjoin results in creating
many paths with cycles that need to be pruned later on. SAA (Algorithm~\ref{algo:SAA}), avoids
creating cycles early on.
\begin{algorithm}
\scriptsize{
\begin{algorithmic}[1]
\caption{\Rank$(E, \len, k)$}
\label{algo:SAA}
\REQUIRE Sorted edge list $E$,
path length $\len$, number of paths $k$
\ENSURE top-$k$ heaviest paths of length $\len$
\STATE $topKBuffer$ = \empty \COMMENT empty sorted set
\STATE $\theta = \len \times W_{max}$ \COMMENT edge weights $\in [0, 1]$
\STATE $ScannedEdges = \empty$ \COMMENT empty hash table
\STATE $e = E.getNextEdge()$
\WHILE{$topKBuffer.BottomScore < \theta$}
\STATE $P^1 = e$
\FOR{$l=1$ to $\len - 1$}
\STATE $P^{i+1} = Join(P^i, ScannedEdges)$ \COMMENT avoid cycles
\ENDFOR
\STATE $Update(topKBuffer, P^{\len})$
\STATE $\theta = e.w + (\len - 1) \times W_{max}$
\STATE $ScannedEdges.put(e.start, e)$
\STATE $ScannedEdges.put(e.end, e)$
\STATE $e = E.getNextEdge()$
\ENDWHILE
\STATE \textbf{return} $topKBuffer$
\end{algorithmic}
}
\end{algorithm}
\Rank algorithm scans the sorted list of edges and reads one edge at a time.
Once a new edge ($e$) is read under sorted access, it is joined with the list
of heavier edges (ScannedEdges) $\ell-1$ times. The $Join$ operation
used in line $8$, takes two input parameters: 1) a list of paths of length $i$, $P^i$;
2) a list of edges, $ScannedEdges$. It extends all paths of length $i$ in $P^i$
with matching edges in $ScannedEdges$ (i.e. smaller search space) and ensures no cycles are created. This
results in paths of length $i+1$. In case $P^{i+1}$ is empty because there are no matching nodes
to join without cycles, naturally, no path will be created. For instance
no path of length $l > 1$ is created after joining the heaviest edge and the empty set
of scanned edges. This is how \emph{Join} is defined as an operation. Join
is one of the operations defined in Relational Algebra. Therefore, it is
not scary to use \emph{Join}, it has been used in practice as part
of the DBMS operations for a long time. We do an enhanced
implementation for self-joins that prunes cycles as well.
Paths of length $\ell$ will be created after enough edges are
scanned in sorted order. This operation is repeated $\ell - 1$ times, until
all possible paths of length $\ell$ including $e$ and heavier edges are
produced. At this point (line 10), \Rank knows that no path can be created during the
rest of execution which is heavier than $\theta = e.w + (\len - 1) \times W_{max}$.
Therefore, if we have already constructed $k$ paths heavier than $\theta$, we
know that these are the correct \topk paths of length $\ell$. In order to
speed up the $Join$ operation, $ScannedEdges$ is maintained as a hash table
that maps nodes (keys) to edges (values). Consider the example graph in Figure~\ref{fig:cliqueGraph}.
Suppose we are interested in the heaviest path of length $3$.
The \Rank algorithm proceeds by scanning the edge list in sorted order of the edge weights.
After reading $4$ edges at depth $d=4$, the edge weight ($w_{4}$) is 0.77 and $\wmax=0.93$.
We can calculate the threshold (upperbound) as: $\theta(2.63) = 0.77 + (3-1)*0.93$.
At depth $5$, \Rank is able to construct the path (6,1,2,3) and $\theta(2.62) = 0.76 + 2 * 0.93$.
Since $\theta$ is equal to the weight of (6,1,2,3), we know it is at least one of
the heaviest paths in the graph and can terminate without processing the rest of the edges
in the sorted list. This is what we call \emph{search space pruning by Sorted Access}.
\section{Limitations of \emph{SAA} and Optimizations}
\label{sec:opt}
\begin{myobs}
\label{obs:lightest}
Let $P$ be a path of length \len\ and suppose $e$ is the lightest
edge on $P$, i.e., its weight is the least among all edges in $P$.
Then until $e$ is seen under sorted access, the path $P$ will not
be constructed by \Rank.
\end{myobs}

Limitation~\ref{obs:lightest}, simply follows the fact that paths
are created using \emph{only} sorted access to the list of edges. If we run
\Rank on the graph of Fig.~\ref{fig:ex1}, we won't create
the heaviest path (a,b,c,d) until the edge (c,d) is scanned. However,
(c,d) is the lightest edge in the graph. This means following
only sorted access, may delay the construction of the heaviest path
until many irrelevant paths are created and discarded. This observation
motivates the following optimization for finding heaviest paths.

\begin{myopt}
We should try to avoid delaying the production of a
heavy path of a certain length until the lightest edge on the path is seen
under sorted access. One possible way of making this happen is via
random access. However, random accesses have to be done with care in
order to keep the overhead low.
\end{myopt}

In the example of Fig.~\ref{fig:ex1}, suppose we
have constructed (a,b,c) and know it is the heaviest
path of length $2$. It may be a good heuristic to
extend (a,b,c) with new edges such as (c,d) and
create paths of length $3$. Since (a,b,c) is the
heaviest path of length $2$, it is very likely to
create heavy paths of length $3$ by extending (a,b,c).
Making access to graph edges such as (c,d), regardless
of where they appear on the sorted list of edges, is
what we refer to as \emph{"\textbf{random access}"}.
An example of random access is presented by Figure 1(b). Making random
access to the edge list simply means reading an edge without caring about where
the edge is located on the sorted list of edges (through the adjacency matrix).
This can help produce
\emph{heavy paths} containing \emph{low weight} edges earlier. Our case study in
Section~\ref{sec:community}, shows one practical example where
a \emph{low weight} edge on a heavy path, can have a meaningful interpretation.
The discussion is provided in Section~\ref{sec:community}, for \emph{Gui-Rong Xue's}
edges with others. Following this optimization, suppose we use random accesses to find ``matching''
edges with which to extend heaviest paths of length $\len - 1$ to length \len.
This way we won't delay the creation of heavy paths until their lightest edge
is visited under sorted access. However, this may not be enough for
early termination of the algorithm. Another decisive factor in early
termination is the upperbound threshold ($\theta$) value on weight of paths we have not
created yet using the \emph{ScannedEdges}.

\begin{myobs}
\label{obs:thresh}
\Rank threshold is conservative compared to what can be easily obtained during the execution.
\end{myobs}

\Rank stops when $\theta$ gets
smaller than the weight of the heaviest path of length $\ell$ discovered so far. Suppose there is an instance in which
the heaviest path of length $\len$ is lighter than $w_{min} + (\len - 1) \times w_{max}$.
In this case, \Rank will produce every path of length $\len$ before reporting
the heaviest path, resulting in no pruning. In Fig.~\ref{fig:cliqueGraph}, this happens
when trying to find the heaviest path of length $4$ and in Fig.~\ref{fig:ex1}, when
trying to find the heaviest path of length $3$. The following lemma
highlights a natural possibility during the execution of \Rank for
obtaining a tighter threshold that results in earlier termination.

\begin{mylem}
\label{theo:nestedThresh}
Let $P$ be the heaviest path of length \len.
When \Rank terminates, every path of length $\len - 1$
that has weight no less than
$P.weight - \wmax$, will be created. This includes the
heaviest path of length $\ell - 1$.
\end{mylem}

\begin{proof}
\scriptsize{
Suppose \Rank finds $P$ at depth $d$. This means
$P.weight \ge w_d + (\len -1) \times \wmax$, and,
$P.weight - \wmax \ge w_d + (\len -2)\times \wmax$. Notice
that $P.weight - \wmax$ is a lower bound on the weight
of the heaviest path of length $\len-1$.
Therefore, if there is a path of length $\len -1$ that is heavier
than $P.weight - \wmax$, it will be produced by \Rank
by depth $d$.}
\end{proof}

One way of making the threshold tighter is by keeping track of shorter paths.
For example, if we know $P$ is a heaviest path of length $\len - 1$,
we can infer that the heaviest path of length \len cannot be heavier than
$P.weight + \wmax$, a bound often much tighter than
$w_d + (\len-1)\wmax$.

For this, we need to keep track of the heaviest paths of length $\len-1$ one by one
to lower the threshold more aggressively, gradually and smoothly.
Pursuing this idea recursively leads to a framework where we maintain
and release heaviest paths of each length $i$, $2 \leq i \leq \len$, at the right time and
make the next threshold tighter. More precisely, we
can perform the following optimization. For this purpose, we can use Buffers as sorted sets (priority queue)
in order to store (insert) and release heavy paths (remove top). \eat{We do not care about the structure
of the buffer and assume it's a sorted set. We use existing libraries for our implementation. Please
also note that the Buffers we present here are simple data structures used by \Main. It is
different than the \emph{memory buffers} discussed in the literature that deals with
low level optimizations of systems.}
\begin{figure}[t]
  \centering
  \includegraphics[scale=0.8, trim = 0mm 0mm 0mm 0mm]{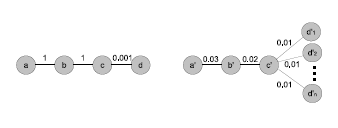}
  \caption{\scriptsize{Example instance of HPP. Graph has one heavy path and $n$ lighter paths.}}
  \label{fig:ex1}
\end{figure}
\begin{figure}[t]
\includegraphics[scale=0.8, trim = 0mm 0mm 0mm 10mm]{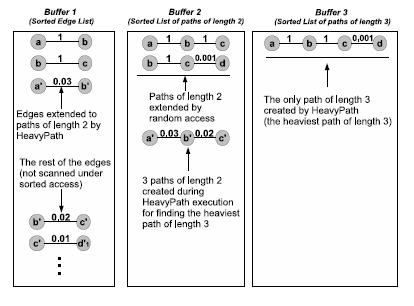}
\caption{\scriptsize{\Main example for finding the heaviest path of length 3 on graph of Figure~\ref{fig:ex1}}}
\label{fig:buffer}
\vspace{-10pt}
\end{figure}
Figure~\ref{fig:buffer} schematically describes the idea of using buffers
for different path lengths. Buffer $B_{1}$ is the sorted edge list for the graph in
Figure~\ref{fig:ex1}. Buffers $B_{2}$ and $B_{3}$ store paths of length 2 and 3 respectively.
When random accesses are performed to extend paths of length $l-1$ to
those of length $l$, the buffers can be used to store these intermediate paths.
For instance, when the heaviest path of length 1 is seen, say edge $(a,b)$ is seen,
it can be extended to paths of length 2 by accessing edges connected to its end points,
as represented by $B_{2}$ in Figure~\ref{fig:buffer}. Similarly, the heaviest
path of length 2 can be extended by an edge using random access to obtain
path(s) in buffer $B_{3}$. Obviously, we do not extend all shorter path segments and extend only
the one that is known to be the next heaviest path of shorter length in its own buffer.
\eat{Having done this, we can also update the threshold($\theta$) for the next buffer.
Thresholds ($\theta$), are updated in sequence for buffers until $\theta_{\ell}$, the
final threshold is reduced.}
We describe this framework in more detail when we present
the final algorithm next.

\begin{myobs}
\label{lim:segment}
\Rank tends to produce
the same sub-path multiple times. For example, when \Rank is required to
produce the heaviest path of length $3$ for the graph in Figure~\ref{fig:ex1},
for paths $(a',b',c',d'_i)$, $i\in[1,n]$, it produces the
length $2$ sub-path $(a',b',c')$ $n$ times as it does not maintain shorter
path segments. This results in more edge reads and significantly increases the
running time.
\end{myobs}

\eat{
Building on the ideas presented so far, we
present the \textbf{\emph{\Main}} algorithm in the next section.
} 
\section{\red{Repeated Sorted Access(RSA) Algorithm}}
\vspace{5pt}
 \label{sec:hp}

We start this section by providing an outline of our main algorithm
for solving HPP.
Our algorithm maintains a buffer $B_i$ for storing paths
of length $i$ explored and not released ($P.w < \theta$), where $2\le i\le \len$.
Let threshold $\theta_i$ denote an upper bound on the
weight of any path of length $i$ that has never been inserted into $B_i$.
\eat{
Each buffer $B_i$ is implemented as a sorted set of paths
of length $i$ sorted according to non-increasing weight
and without duplication
(e.g., paths $(a,b, c)$ and $(c,b,a)$ are considered duplicates).

Each buffer is also a set and does not allow duplicates.
Sorted set is one of the data structures provided
by Java programming language and cost of insertion
in a sorted set is $log(n)$ where $n$ is the number of paths
existing in the buffer.
The main function is provided by Algorithm~\ref{algo:main}.
Every buffer is initialized as an empty sorted set.
Each threshold is initialized to the maximum possible value at start.
}%
Algorithm~\ref{algo:main} describes the overall approach.
It takes as input a list of edges $E$ sorted in non-increasing order of edge weights,
\red{\erase{a graph $G$,}} and parameters \len and $k$.
It calls the \RJ method (Algorithm~\ref{algo:HP})
repeatedly until the \topk heaviest paths of length \len
are found.

\begin{algorithm}[t]
\scriptsize{
\begin{algorithmic}[1]
\caption{\Main$(E, \red{\erase{G,}} \len, k)$}
\label{algo:main}
\REQUIRE Sorted edge list $E$, 
path length $\len$, number of paths $k$
\ENSURE top-$k$ heaviest paths of length \len \red{\erase{in $G$}}
\FOR{$l=2$ to \len}
\STATE $B_{l}\gets\emptyset$  \COMMENT empty sorted set 
\STATE $\theta_l = \wmax \times l$
\ENDFOR
\STATE $topPaths\gets\emptyset$ \COMMENT empty sorted set
\WHILE{$\mid topPaths\mid < k$}
\STATE $topPaths\gets topPaths$ $\cup$ \RJ$(E,\red{\erase{G,}}\len)$
\ENDWHILE
\end{algorithmic}
}
\end{algorithm}
\begin{algorithm}[t]
\scriptsize{
\begin{algorithmic}[1]
\caption{\RJ$(E, \red{\erase{G,}} l)$}
\label{algo:HP}
\REQUIRE  Sorted list of edges $E$, \red{\erase{a graph $G$,}}  and
path length $l$
\ENSURE Next heaviest path of length $l$
\IF{$l = 1$}
\STATE $P^{1} \gets$ {\sc ReadEdge}$(E)$
\STATE $\theta_2 = 2 \times P^{1}.weight$
\STATE \textbf{return} $P^{1}$
\ENDIF
\WHILE{$B_{l}.topScore \le \theta_{l}$}
\STATE $P^{l-1} \gets$ \RJ$(E,\red{\erase{G,}} l-1)$	\COMMENT recursion
\STATE $s,t \gets$ {\sc EndNodes}$(P^{l-1})$
\FORALL{$y\in V\mid(y,s) \in E$}
\STATE $B_{l}\gets B_{l}\cup((y,s)+P^{l-1})$  \COMMENT avoiding cycles
\ENDFOR
\FORALL{$z\in V\mid(t,z) \in E$}
\STATE $B_{l}\gets B_{l}\cup(P^{l-1}+(t,z))$  \COMMENT avoiding cycles
\ENDFOR
\ENDWHILE
\STATE $P^{l} \gets$ {\sc RemoveTopPath}$(B_{l})$
\IF{$l < \len$}
\STATE $\theta_{l+1} = \max(B_{l}.topScore, \theta_l) + \wmax$
\ENDIF
\STATE \textbf{return} $P^{l}$
\end{algorithmic}
}
\end{algorithm}

\eat{
\IF{$B_{l}.topScore > \theta_{l}$}
\STATE $X = B_{l}.removeTopPath()$
\IF{$l < \len$}
\STATE $\theta_{l+1} = MAX(B_{l}.topScore, \theta_l) + W_{max}$
\ENDIF
\STATE \textbf{return} $x$
\ENDIF
\STATE $nextPath_{l-1} = \RJ(E,l-1,G)$
\STATE $s = nextPath_{l-1}.StartNode$
\STATE $t = nextPath_{l-1}.EndNode$
\FORALL{nodes $z$ s.t. $(t,z) \in G$ avoiding cycles}
\STATE $B_{l}.insert(nextPath_{l-1}+t\leadsto z)$
\ENDFOR
\FORALL{nodes $z'$ s.t. $(z',s) \in G$ avoiding cycles}
\STATE  $B_{l}.insert(z'\leadsto s+nextPath_{l-1})$
\ENDFOR
\STATE $ \RJ(E,G,l)$
}%
%
Algorithm~\ref{algo:HP}\footnote{\scriptsize{In order to avoid StackOverflow we do a non-recursive implementation. We use
the recursive pseudo code because it is more intuitive.}} describes the \RJ method.
It takes as input a list of edges $E$ sorted in non-increasing
order of edge weights,\red{\erase{a graph $G$,}}  and
the desired path length $l$, \red{$l \ge 2$}.
It is a recursive algorithm that produces heaviest paths of shorter lengths
on demand, and extends them with edges to produce paths of length \len.
The base case for this recursion is when $l = 1$ and
the algorithm reads the next edge from the sorted list of edges.
The {\sc ReadEdge} method 
returns the heaviest unseen edge in $E$ \red{(sorted
access, line 2)}.
If $l < \len$, the path of length $l$ obtained as a result of the recursion
is extended by one hop 
to produce paths of length $l + 1$.
Specifically, a path of length $l<\len$ is extended using edges
\red{(random access, lines 8 and 10)}
that can be appended to either one of its ends (returned by method
{\sc EndNodes}).  The ``+'' operator for appending an edge
to a path is defined in a way that guarantees no cycles are created.
The threshold $\theta_l$ is updated aggressively when the next
heaviest path of length $l-1$ is released from $B_{l-1}$. 
This is done by calling the method {\sc RemoveTopPath}
for buffer $B_{l-1}$ and returning the resulting path. If there
is no path already in the buffer that beats $\theta_{l-1}$, this results
in more recursion using smaller $l$, until this condition
becomes true at some point during the execution. At any point during
the execution
if $i < \len$ and the next heaviest path of length $i$ (P) has been obtained
($B_i.topScore > \theta_i$), $\theta_{i+1}$ is more intuitively updated as $\theta_{i+1} = P.W + W_{max}$~\footnote{\scriptsize{We use a bound tighter threshold, line $14$ of Algorithm $3$.}}.

\begin{mytheo}
\label{algo:HPcorrectness}
Algorithm \Main correctly finds top-$k$ heaviest paths of length \len.
\end{mytheo}
\begin{proof}
\scriptsize{
The proof is by induction. The base case is for going from edges
to paths of length $2$. Given that all of the edges above depth $d$
are extended, the heaviest path that can be created from them is already
in $B_2$. The weight of the heaviest path that can be created from lighter edges is
at most $2 \times w_d$.
If the heaviest path in $B_2$ is heavier than $2 \times w_d$, then it must be the
heaviest path of length $2$. Assuming the heaviest
paths of length $l$ are produced correctly in sorted order, we show the heaviest path of length
$l+1$ is found correctly.
Suppose $P$, the heaviest path of length $l+1$, is created for the first time by
extending $Q$, which is the $n^{\rm th}$ heaviest path of length $l$.
The next heaviest path of length $l$ is either already in
$B_{l}$ or has not been created yet.
Therefore, $\max(\theta_{l}, B_l.topScore)$ is an upper bound
on the next heaviest path of length $l$ that has not been identified yet. Any
such path can be extended by an edge of weight at most $\wmax$.
Suppose when the $m^{\rm th}$ heaviest path of length
$l$ is seen, $\max(\theta_{l}, B_l.topScore) + \wmax$ is updated to a value
smaller than $P.weight$. It is guaranteed that $P$ is already
in $B_{l+1}$ and has the highest weight in that buffer.
In other words, when the threshold is smaller than $P.weight$,
the difference between the weight of $P$ and next heaviest path of length $l$
is more than $\wmax$. Now,
paths of length $l+1$ that can be created from heavier paths of
length $l$ are already in the buffer, and no unseen path of length $l$ can
be extended to create a path heavier than $P$.
Therefore, $P$ is guaranteed to be the heaviest path.
The preceding arguments hold for \topk heaviest paths where $k>1$.
}
\end{proof}

\section{Community Detection Case-Study}
\label{sec:community}
\vspace{5pt}
\begin{figure}[t]
\label{Communities}
\subfigure[Community 1 (Chinese University of Hong Kong, Aliyun.com, MSRA)]{
\includegraphics[scale=0.7]{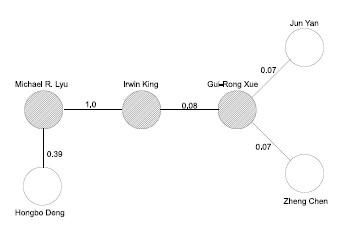}
\vspace{70pt}
\label{fig:depth1}
\vspace{70pt}
}
\subfigure[Community 2 (Yahoo!, Google)]{
\includegraphics[scale=0.7]{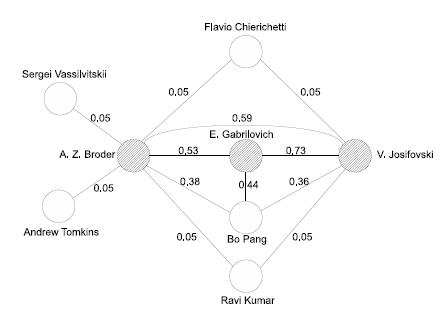}
\label{fig:depth2}
}\label{fig:access}
\vspace{-10pt}
\caption{\scriptsize{Two sample communities along with their cores, found using parameters $\ell = 5$, $k = 100$ and support = $10$.}}
\vspace{-10pt}
\end{figure}

Suppose we have a co-authorship graph that connects researchers
according to their publications. The edge weight is defined
as the sum of "\emph{pairwise collaboration credits}"
that researchers get from each paper they publish together.
Pairwise collaboration credit for a pair of authors in a paper,
is $1/authorNum$. Where $authorNum$ is the number of authors
of a paper. Given a collection of publications, we define an
edge between two researchers if they have published at least one paper
together. Then use the sum of "\emph{pairwise collaboration credits}"
they get from all of the papers they have published together as the
edge weight. Single author papers do not contribute to edge weights.
We further normalize edge weights by the maximum
so that they are all in $[0,1]$. Using this graph, and given parameters
\len and $k$ and a \emph{popularity}~\footnote{\scriptsize{Same as support in frequent pattern mining}} parameter, we do the following,

\begin{enumerate}
  \item Find top-$k$ heaviest paths of length $l$
  \item Finding \textbf{\emph{cores}} of communities: Find all \textbf{\emph{maximal}} frequent sub-paths
  in top-$k$ paths. Frequent
  sub-paths are those that appear in more number of paths than the \emph{popularity} parameter.
  \item Aggregate top-$k$ paths and form the \emph{aggregate graph of top paths}
  \item Given a \textbf{\emph{core}} and the \emph{aggregate graph of top paths}, do the following: add any edges that connect
  any other researchers in the aggregate graph of top paths to the core. This
  shapes a community of \emph{core} researchers that appear on many heavy paths. It also
  adds those other researchers who are strongly connected to the core community.
  \item Communities can grow using larger values for $k$ that adds more nodes to the picture. In a way, $k$ can
  be used to zoom in and zoom out of communities.
\end{enumerate}

Using DBLP data~\cite{DD}, we constructed this graph using papers published
between $2008$ and $2011$ (inclusive), presented in the main program one of these conferences: "SDM",
"PKDD", "ICDM",
"SIGKDD", "CIKM", "SIGIR", "ICML", "NIPS", "WWW" and "WSDM". Figure 4, shows
two of the communities we find using parameters $l=5$, $k=100$ and $popularity = 10$.
We provide this graph (DBLPG), and the top-$100$ paths of length $5$, to enable
the reader to reproduce the communities described~\footnote{\scriptsize{http://ucalgary.ca/{\raise.17ex\hbox{$\scriptstyle\sim$}}mkhabbaz/DBLPG.zip}}.
We notice that the core of community 1 (Figure~\ref{fig:depth1}) contains the pair of researchers
with the highest edge weight in the graph. This shows a long lasting
collaboration. We did a quick google search about each of the core
communities in DBLPG. Irwin King and Michael R. Lyu work for the Chinese
University of Hong Kong. They are connected to Gui-Rong Xue who
is the senior director at Aliyun.com and he is connected to a principal
researcher from Microsoft research Asia (Zheng Chen).

Most of the researchers in community 2 (Figure~\ref{fig:depth1}) have worked for Yahoo! Research
Labs, including the three researchers in the core. Andrei Z. Broder is
currently a distinguished scientist at Google. He has previously worked at Yahoo!, IBM and
Altavista. We also found on his Wikipedia page that
he is the inventor of MinHash locality sensitive hashing. Evgeniy Gabrilovich, is a
senior staff research scientist at Google who also used to have a central role
at Yahoo! research. Vanja Josifovski apparently still works for Yahoo! and is a
Principal Research Scientist and the Lead of the Performance Advertising Group at Yahoo! Research.
Most of the other researchers in community 2 work for yahoo! as their emails appear
on their publications. Some may have moved to Google recently!

It is worth noting that we discovered three other cores as well in our output.
Due to space limitations, we only report the cores:
\emph{core 3} =  $\{$Andrei Z. Broder, Evgeniy Gabrilovich, Xuerui Wang $\}$(weight = 0.73),
\emph{core 4} = $\{$Shuicheng Yan, Ning Liu, Jun Yan, Zheng Chen $\}$ and
\emph{core 5} = $\{$Ning Liu, Jun Yan, Zheng Chen, Weizhu Chen $\}$.
Xuerui Wang is also a scientist at Yahoo! labs. Cores $4$ and $5$ have size $4$.
We notice that a sub-path of length $3$ is
shared between these two, with all researchers from Microsoft Research Asia.
Community $1$ is connected to two of these three researchers through
Gui-Rong Xue. This shows a potential for collaboration between these two
\emph{cores} in the future! who are the Chinese University of
Hong Kong and Microsoft Research Asia. Changing the parameters
$l$, $k$ and \emph{popularity}, changes the output in the form
presented in Figure~\ref{fig:access}. It can also help with creating
a visualization tool that adapts its output with the current parameter
values and be used for visualizing hierarchies and relationships
of communities at different scales. Such a tool needs to run \Main repeatedly
with different parameter values in order to change its view. This highlights
the importance of scalability and efficiency for HPP problem.
We leave the continuation and
implementation of this idea to our future work.

The case study presented in this section, is also an example
of an application where finding the exact solution is required.
This is because we focus on $100$ heaviest paths and frequent
pattern mining in order to discover the cores of communities.
Heuristic or any sort of approximate solutions may result in paths that
do not necessarily go through the cores that form the
top communities. \emph{Top-$k$ heaviest paths, rely on their cores to be heavy!}

\section{Experimental Analysis}
\label{sec:expn}
\subsection{Reproducibility}
\label{sec:data}
We use four graphs in our experiments. We create graphs using Cora\footnote{\scriptsize{\url{http://www.cs.umass.edu/~mccallum/data}}},
last.fm\footnote{www.last.fm} and DBLP datasets. \textbf{Cora(70 nodes-1580 edges):} nodes represent research topics and edge weights are defined using the average fraction of citations between the two topics. \textbf{Last.fm(40k nodes-183k edges)}: we crawled the existing collection of playlists. Nodes represent songs and edges represent co-occurrence in playlists. Edge
weight is defined using the Dice coefficient\footnote{\scriptsize{$\texttt{dice}(i,j) = \frac{2|i\cap j|}{|i|+|j|}$}}. \textbf{BAY(321k nodes-400k edges)}: it represents the road network in San Francisco Bay area\footnote{\scriptsize{\url{http://www.dis.uniroma1.it/}}}. For this graph we perform the $C-W$ (C=1) transformation described in Section~\ref{sec:probdef}, and solve $\ell$-TSP because distance minimization
makes more sense for road networks. \textbf{DBLPG(4.5k nodes-9k edges)}: Section~\ref{sec:community} describes how
the graph is created. In all cases, the graph we work with has edge weights in $[0-1]$ and the input graph
is a text file that stores one edge per line in this format: "source destination weight". You can
access a .zip file containing all the preprocessed graphs\footnote{\scriptsize{http://ucalgary.ca/{\raise.17ex\hbox{$\scriptstyle\sim$}}mkhabbaz/Graphs.zip}}.
We do not use graphs with millions of edges since the problem is NP-hard and we focus on finding exact answers.
On the other hand, all of the graphs we use in experiments deal with practical real life applications.
One interesting fact we discovered in our experiments is that the smallest of the graphs (Cora)
we use, is one of the most challenging scenarios. This is due to the fact that its average node degree is higher
than others and this results in an exponentially larger number of longer paths when $\ell$ increases. We do not report
any results regarding the DFS algorithm because it is inferior to \DP in all cases. Source code of all algorithms, along with instructions
for running can be downloaded\footnote{\scriptsize{http://ucalgary.ca/{\raise.17ex\hbox{$\scriptstyle\sim$}}mkhabbaz/HPPCase.zip}}. All experiments were performed on a Linux machine with 64GB of main memory and 2.93GHz-8Mb Cache CPU. To be consistent, we allocated 12GB of memory for each run.
In all of the cases where \Main\footnote{\scriptsize{We use the name HeavyPath for our implementation of the algorithm.}} runs out of memory, other algorithms either also run out of the 12GB allocated
memory, or do not terminate after running for a couple of days. \eat{We do not have any estimate of how much longer
it takes until DFS manages to find the exact solution for larger $\ell$ on some graphs. In theory, it can
terminate after life is over on earth!} Our implementation is in Java. For all algorithms other
than \Main, we provide implementation such that the time spent on Garbage Collection can be
measured accurately using the \emph{Jstat} tool\footnote{\scriptsize{Java does garbage collection automatically using
GC algorithms. http://docs.oracle.com/javase/1.5.0/docs/tooldocs/share/jstat.html}}. All of the figures that present running time use log scale on Y-axis. We further
analyzed GC time and found it is negligible compared to the total running time in all cases as supposed to be
since Java is one of the most reliable programming languages.

\subsection{Empirical Evaluation}
\begin{figure*}[t]
\centering
\subfigure[Vary $\len, k=1$, Cora ]{
\includegraphics[width=\figthree]{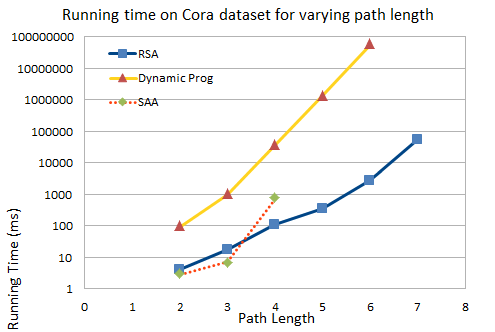}
\label{fig:Lcora}
}
\subfigure[Vary $\len, k=1$, last.fm ]{
\includegraphics[width=\figthree]{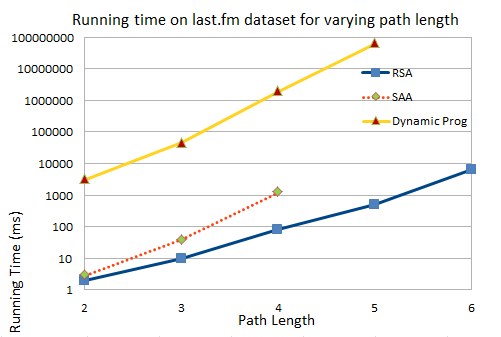}
\label{fig:Llastfm}
}
\subfigure[Vary $\len, k=1$, Bay ]{
\includegraphics[width=\figthree]{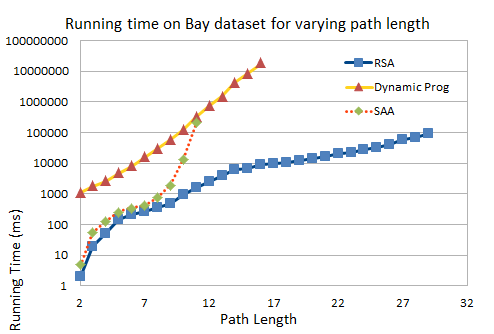}
\label{fig:Lbay}
}
\subfigure[$\len=4$, vary $k$, Cora ]{
\includegraphics[width=\figthree]{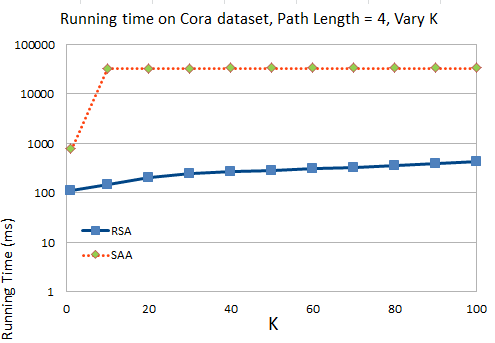}
\label{fig:kcora}
}
\subfigure[$\len=4$, vary $k$, last.fm ]{
\includegraphics[width=\figthree]{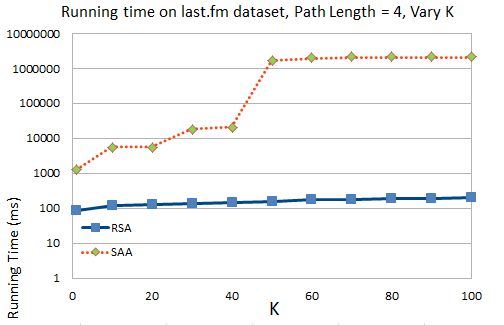}
\label{fig:klastfm}
}
\subfigure[$\len=10$, vary $k$, Bay ]{
\includegraphics[width=\figthree]{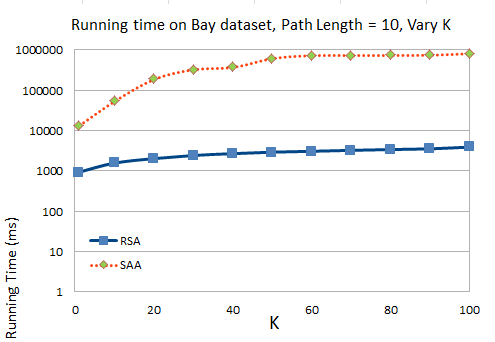}
\label{fig:kbay}
} \vspace{-10pt}
\caption{\scriptsize{Running time comparisons for exact algorithms with different parameter settings.}}
\vspace{-10pt}
\label{fig:exact}
\end{figure*}

\begin{figure*}[t]
\centering
\subfigure[Vary $\len, k=1$, Cora ]{
\includegraphics[width=\figthree]{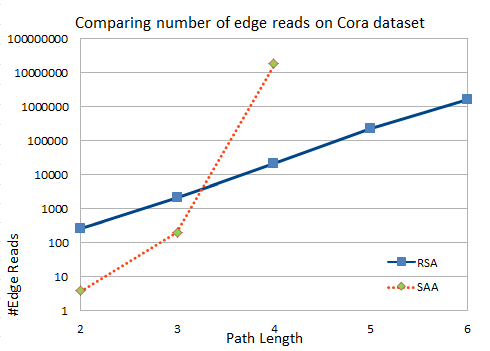}
\label{fig:Ecora}
}
\subfigure[Vary $\len, k=1$, last.fm ]{
\includegraphics[width=\figthree]{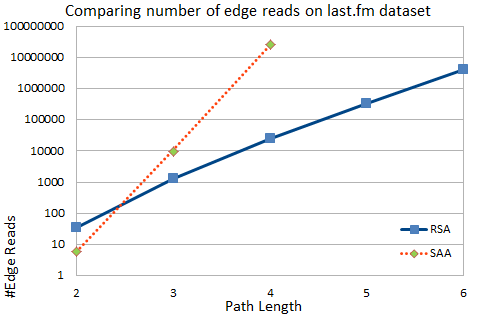}
\label{fig:Elastfm}
}
\subfigure[Vary $\len, k=1$, Bay ]{
\includegraphics[width=\figthree]{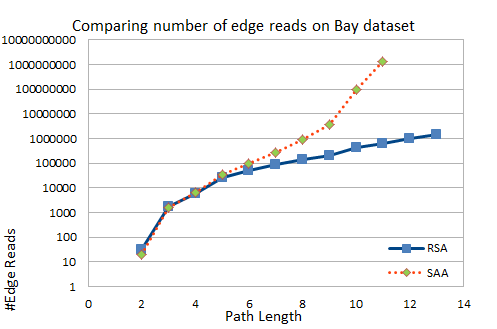}
\label{fig:Ebay}
}
\vspace{-10pt}
\caption{\scriptsize{Comparing the number of edge reads between SAA and \Main for different \len}}
\vspace{-10pt}
\label{fig:EdgeReads}
\end{figure*}

\begin{figure*}[t]
\centering
\subfigure[Vary $\len, k=1$, DBLPG: Running Time(ms) ]{
\includegraphics[width=\figfour]{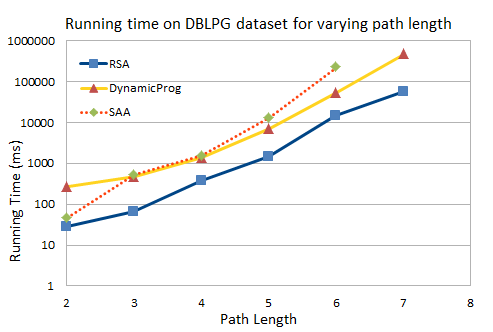}
\label{fig:LDBLPG}
}
\subfigure[Vary $\len, k=1$, DBLPG: Number of Edge Reads]{
\includegraphics[width=\figwidth]{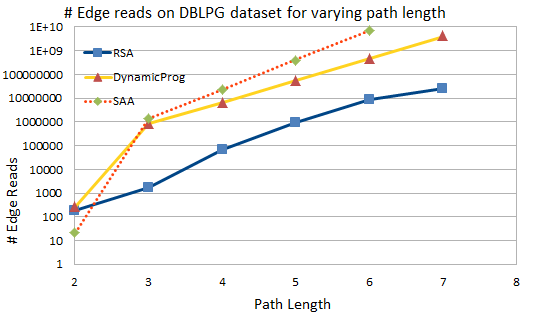}
\label{fig:EDBLPG}
}
\vspace{-10pt}
\caption{\scriptsize{Scalability results on DBLPG dataset}}
\vspace{-10pt}
\label{fig:DBLPResults}
\end{figure*}

Figure~\ref{fig:exact}(a-c) show the running time for
finding the top-1 path of various lengths for Cora, last.fm and Bay graphs.
Average node degree plays a key role in the empirical hardness of HPP. This is due to the fact that
the complexity of the problem increases exponentially with node degree.

The running time increases with $\ell$ for all algorithms, as supposed to be.
\Rank does not manage to go beyond short path lengths in most cases, the reasons
being those limitations mentioned in Section~\ref{sec:opt}. Beyond some depth
down the sorted edge list, \Rank becomes inefficient due to the increasing
complexity of the problem and the fact that it either does not manage
to construct heavy paths early on or the fact that the threshold decays slowly.
For shorter path lengths \Rank terminates earlier than \DP because it does
pruning and uses a threshold for termination while \DP explores a larger
search space. \Main performs considerably more efficient than other algorithms
except for one case on Cora graph for $\ell = 3$ where it is slightly less
efficient than SAA. We believe this is due to the fact that \Main uses a combination
of sorted and random access. Random access results in fundamental improvements
in the running time as we see in most cases. However, in some rare cases \Rank
may get lucky because it avoids random access and also the edge weights
are such that result in fast termination. Albeit, at the cost of losing
all the other experiments. All in all, our findings show that
\Main is a way more scalable and reliable algorithm. \DP behaves consistently and manages to achieve
pruning by aggregating shorter path segments early on. However, it finds
the heaviest path ending at every node and although it scales to longer lengths,
it is orders of magnitude slower than \Main. We would like to highlight again
that we use log scale on Y-axis and the difference in the running times is
considerable. All in all, there is no question about the surprising scalability of \Main
in all experiments.

Figures~\ref{fig:exact}(d-f), compare the running times of \Rank and \Main
for varying $k$ and fixed $\ell=4$, since these two are the algorithms
designed to work as top-$k$ algorithms. In all cases, \Main is more efficient and
it spends a smaller marginal running time compared to \Rank for finding top-$k$.
We observe
sudden jumps in the running time of \Rank when $k$ changes while \Main continues
to work reliably for values of $k$ less than $100$ in these experiments and
scales more smoothly.

We further compare the number of edge reads of \Rank and \Main in Figure~\ref{fig:EdgeReads},
that is a system independent notion of the running time. We observe patterns very similar to
Figures~\ref{fig:exact}(a-c). This highlights the fact that the number of edges read during
the execution is the main factor determining the running time. We observe in some cases (for small $\ell$),
\Rank reads fewer edges but spends slightly more time for execution. This is due to the
fact that \Main is performing random access using the adjacency matrix of the graph to construct
longer paths. However, \Rank performs the $Join$ operation and looks up edges in a hash
table which itself requires few milliseconds of time overall.

\begin{figure}[t]
\label{DBHeavyPaths}
\includegraphics[scale=0.7]{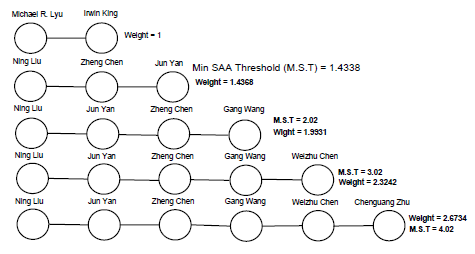}
\caption{\scriptsize{Heaviest paths of $\ell = 1$ to $\ell = 5$ on DBLPG Graph. SAA terminates for $\ell = 2$ before processing all edges at
depth $28$ (out of 9060 edges). The minimum threshold value when SAA terminates is 1.4338 and smaller
than the weight of the heaviest path of $\ell = 2 (1.4388)$. For $\ell > 2$, SAA can not terminate without processing
all edges i.e. almost equivalent to the DFS algorithm. As can be seen in Figure $7$, SAA terminates earlier than \DP for
$\ell = 2$ but when $\ell > 2$, \DP becomes more efficient since SAA behaves as naive as DFS.}}
\vspace{-10pt}
\end{figure}

Finally, Figure~\ref{fig:DBLPResults} shows the running times and the number
of edge reads on the DBLPG graph for finding the heaviest paths of different
lengths. We notice on this Graph, SAA algorithm performs more poorly compared
with the Dynamic Programming algorithm. Of course, \Main beats all algorithms in all
cases and finds the exact solution almost $10$ times faster than the rest. \eat{ SAA performs
better than the Dynamic Programming algorithm for $\ell = 2$ because it can terminate
by initiating the search using only $28$ edges. Depth $28$ is where $\theta$ becomes smaller than the weight
of the heaviest path of $\ell = 2$. For $\ell > 2$, the weight of the heaviest path is smaller
than $(\ell-1)\times w_{max}+w_{min}$ as highlighted by Limitation~\ref{obs:thresh}. In these
cases, the SAA algorithm becomes equivalent to the DFS algorithm and is less efficient than the
Dynamic Programming algorithm that boosts performance by aggregating shorter path segments early on.
This highlights the importance of investigating smarter strategies for updating $\theta$ and using
a tighter threshold which is done by our optimizations presented in Section~\ref{sec:opt}, that
led to designing the \Main algorithm.} Figures 7(b) and 8, illustrate this.

\section{Discussions}
\label{sec:discussion}
Our theoretical results prove that HPP is np-hard and inapproximable in the general case. While
it makes sense to seek more tractable problem definitions by adding constraints, still we show experimentally
that we can achieve scalability practical enough for our main algorithm to be used in a real time fashion and
be used in smart software technologies. We use a \emph{branch and bound} solution for finding the exact solution, inspired
by top-k algorithms. We use a constraint $\ell$, on edge weights and this makes our comparison for finding heavy paths,
representing significantly important sequences in the graph more fair. Comparing path weights of paths with different
lengths may not make enough sense in the \emph{\emph{general case}}. We show exact solution even for small $\ell$, if obtained, can be used in practice effectively, while other algorithms fail to report even for small $\ell$ in a timely manner. It is
obviously nice to scale to longer paths and this requires smarter strategies for reducing $\theta$ during the execution
that leads to faster search space pruning. We make a proposal for designing a network visualization tool-box with
zoom-in and zoom-out functionalities with changing $k$. We find our presented case study practical but obviously as
mentioned in Section~\ref{sec:intro}, applications are not restricted to this and there are other applications such as
those in bioinformatics. We can use high probability sequences in order to fix technology related errors such as replacing inaccurate or suspicious entries or filling in missing values in biological sequences. We can also use these high probability sequences for feature extraction from DNA which is an extremely long sequence of symbols in bioinformatics for classification
and other learning tasks. One reason we choose visualizing bibliographic networks is presenting interpretable results
for computer science audience. While there is a variety of solutions to be used in these applications, we want to
propose a solution to a well-defined problem in computer science similar to TSP, that can be significantly leveraged
for creating technologies in software industry. Another way to add constraints to the problem and make pruning
more possible is through designing reliable heuristic algorithms with constraints on the distribution of edge weights. Of course,
 this requires careful testing for estimating distributions and we can focus on this extension in the future.

\begin{figure}[t]
\label{PathWeights}
\includegraphics[scale=0.6]{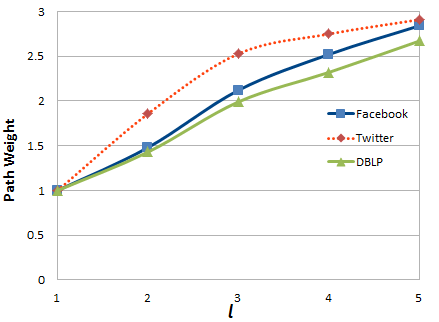}
\caption{\scriptsize{Heaviest path weights of $\ell = 1$ to $\ell = 5$ for Twitter, Facebook and DBLPG using Egonetwork circles.}}
\vspace{-10pt}
\end{figure}

 Recently, there have been advancements in the field of social networks described as Egonetworks~\footnote{\scriptsize{\url{http://snap.stanford.edu/data/index.html}}}. Egonetworks are typically
 presented as groups of people involved in different \emph{"circles"}. Our bibliographic network described in Section~\ref{sec:community},
 is similar to the notion of circles in Egonetwork where each list of authors represents a circle. In an extreme scenario
 there may be many nodes involved in each circle in domains such as concrete social networks. Any problem definition for discovering dense or significant subgraphs according to such networks can result in finding significant sub-structures. We propose an effective approach for constructing such a graph. Based on the data available we prefer to make no judgement about the significance of individuals. Our proposed solution based on finding top-k heavy paths is proven to be scalable by experiments. On the other hands we can manage to change the output using simple parameters such as $k$, that enables us to zoom-in and zoom-out. Our argument for reliability of our output is the following observation. We notice many heavy paths share the same exact shorter sub-paths. This emphasizes the importance of some collaborations in the bibliographic network around which an important and successful research community is formed. This led to our definition of community detection approach that makes use of maximal frequent subpaths for discovering such frequent short paths shared among many top-k heavy paths. We follow the same approach for
 constructing graphs from Facebook and Twitter Egonetworks. In particular, we only use circles and convert mutual presence of
 nodes in circles to define edges. We find the same exact method described in Section~\ref{sec:community}, in order to define
 edge weights. Figure~\ref{PathWeights}, reports the heaviest path weights of $\ell = 1$ to $\ell = 5$, and compares them with
 those obtained from DBLP data. We use the union of all circles first and then construct the graph. Since edge weights are all scaled to $(0-1)$, we expect heaviest paths of different length to be in the same range for all datasets. In comparison, we notice in Figure~9, that all graphs follow a similar curve for increasing $\ell$. We notice that the slope decreases with $\ell$, and it is quite likely path weight converges quickly for large $\ell$. This increases the chances of finding reliable heuristic algorithms with output close enough to the exact answer of the np-hard problem. Facebook and Twitter datasets result in paths with slightly heavier weights and we believe this is because in real life networks people with strong connections belong to many active networks while researchers may not be as socially active as ordinary people. Although the difference in path weights is small and almost negligible. The most notable facts to highlight are the similar shape of curves as well as the decreasing slope.

\section{Related work}
\label{sec:related}
\eat{We already mentioned some related work and application
domains that are relevant to \hpp in Section~\ref{sec:intro}.
In this section
we discuss those related work that have a similar
flavour to \hpp in their technical contribution and motivation.}
In \cite{Hansen:09}, Hansen and Golbeck make a case for
recommending playlists and other collections of items.
The AutoDJ system of Platt et al.~\cite{Platt:01} is one of the early
works on playlist generation
that uses acoustic features
of songs rather than a graph based representation. Random walk
methods were used in~\cite{Baluja:08}, in order to generate lists
of videos in a co-view graph. We formalized and presented our
view of generating lists using random walks in introduction for
recommending lists as an ordered package of items to users.

The \hpp in a special case was studied
recently by~\cite{Bansal:07} and they defined a notion of \emph{Stable Clusters}
for the specific application of finding persistent
chatter in the blogosphere. \eat{We generalize the problem and focus on
more sophisticated algorithms. We did not focus on clustering applications
of \hpp in this work but this is something worth investigating in the future.}
Unlike our setting, the graph associated with
their application is \len-partite and acyclic.
Absence of cycles makes their technical problem
more tractable than ours. The most efficient
algorithm presented in their work is a BFS
based method that follows a dynamic programming
approach. \eat{It is not straight forward to adapt their
algorithm for \hpp in the general case to produce paths
without cycles.} We followed dynamic programming algorithm that also avoids cycles and found this
approach extremely inefficient
compared to \Main.

Rankjoin was first proposed
by Ilyas et al.~\cite{Ilyas:03, Ilyas:04, Ilyas:05, Finger:2010, RA:11} to produce top-$k$ join tuples in
a relational database. We show how to convert Rankjoin to HPP. \eat{We identified its relevance to \hpp. More specifically,
it solves the problem on acyclic $(\len+1)$-partite graphs. Most of the literature
on Rankjoin focuses on $\ell=2$ which is not NP-hard. Furthermore, most
of the effort is spent on minimizing the number of block reads from disk and
less attention has been paid to Rankjoin from an algorithmic point of view.}
Furthermore, no one has studied Rankjoin with self-joins to the best of our knowledge.
In~\cite{Schnaitter:10}, authors provide a comprehensive summary of this algorithm and
do optimality analysis as done in the theoretical database literature. Their conclusion
is that beyond $\ell=2$, we can not guarantee any optimality results for Rankjoin. This
also matches our results regarding HPP being np-hard and inapproximable.

\section{Conclusions and Future Work}
\label{sec:concl}
Finding the \topk heaviest paths in a graph is an
interesting problem with many practical applications. We discuss the hardness
of this problem. We focus on developing
practical exact algorithms for this problem.
We use an innovative top-$k$ query processing algorithm.
We motivate the problem from several different perspectives
and discuss possible applications. We present a case study on
core community detection. Our findings suggest that our case
study can be further extended to a community detection tool
that discovers the key influential people in a network represented
by a weighted graph. Our future work will focus on more applications
of \Main, specially those in network analysis, community detection, recommender systems and bioinformatics.
We would also
like to investigate more scalable algorithms and threshold
update strategies within the same framework. We also intend
to design probabilistic and robust heuristic algorithms that
work under memory and other types of constraints such as edge distribution.

\end{document}